\documentclass[11pt,reqno]{article}
\pdfoutput=1
\usepackage[utf8]{inputenc}
\usepackage{marvosym}
\usepackage{amsmath, amssymb, amsfonts, bm}
\renewcommand{\mathbf}[1]{\boldsymbol{#1}}

\usepackage{amsthm}
\usepackage[colorlinks=true,citecolor=blue]{hyperref}
\usepackage[noadjust]{cite}
\usepackage{libertine}
\usepackage[T1]{fontenc}
\usepackage{microtype}
\usepackage{geometry}
\usepackage{fullpage}
\newtheorem{theorem}{Theorem}

\newtheorem{lemma}[theorem]{Lemma}

\theoremstyle{definition}

\newcommand{\abs}[1]{\ensuremath{\left|#1\right|}}

\newcommand{\defeq}[0]{\ensuremath{:=}}

\newcommand{\pdiff}[2][]{\ensuremath{\frac{\partial#1}{\partial#2}}}

\newcommand{\inb}[1]{\ensuremath{\left\{#1\right\}}}
\newcommand{\inp}[1]{\left(#1\right)}

\renewcommand{\deg}[1]{\ensuremath{\mathrm{deg}\inp{#1}}}

\renewcommand{\vec}[1]{\mathbf{#1}}
\newcommand{\nvec}[2][V]{\inp{#2_v}_{v \in #1}}

\newcommand{\D}[1]{\ensuremath{\mathcal{D}_{#1}}}
\newcommand{\ignore}[1]{}

\author{
         \textsc{Piyush Srivastava}\footnote{UC Berkeley, \texttt{piyushsriva@gmail.com}} \and
        \textsc{Mario Szegedy}\footnote{Rutgers University, \texttt{szegedy@cs.rutgers.edu}}
        }

\begin{document}
\title{A simplified proof of a Lee-Yang type theorem}
\maketitle
In this short note, we give a simple proof of a Lee-Yang type theorem
which appeared
in~\cite{sinclair14:_lee_yang_theor_compl_comput_averag}.  Given an
undirected graph $G = (V, E)$, we denote the partition function of the
(ferromagnetic) Ising model as
\begin{displaymath}
  Z(G, \beta, \vec{z}) \defeq \sum_{\sigma: V \rightarrow \inb{+,
      -}}\beta^{d(\sigma)} \prod_{v:\sigma(v) = +} z_v, \;\;\;\;\;\;\;\;(\vec{z}= \nvec{z})
\end{displaymath}
where $d(\sigma)$ is the number of edges $e = \inb{i, j}$ such that
$\sigma(i) \neq \sigma(j)$, and $0 < \beta < 1$ is the \emph{edge
  activity}.  The arguments $z_i$ of the partition function are called
\emph{vertex activities} or \emph{fugacities}. We then define the
operator $\D{G}$
\begin{displaymath}
  \D{G} = \sum_{v \in V} z_v \pdiff[]{z_v},
\end{displaymath}
which derives its importance from the fact that the \emph{mean
  magnetization} $M(G, \beta, \vec{z})$ of the Ising model on $G$ for
a given setting of the edge activity and the fugacities can be written
as
\begin{displaymath}
  M(G, \beta, \vec{z}) = \frac{\D{G}Z(G, \beta, \vec{z})}{Z(G,
    \beta, \vec{z})}.
\end{displaymath}
The theorem whose proof in this note we simplify is the following:
\begin{theorem}[\cite{sinclair14:_lee_yang_theor_compl_comput_averag}]
  \label{thm:1}
  Let $G = (V,E)$ be a connected undirected graph on $n$ vertices, and
  assume $0 < \beta < 1$.  Then $\D{G}Z(G, \beta, \vec{z}) \neq 0$ if
  for all $v \in V$, $z_{v}$ is a complex number with absolute value
  one.
  \end{theorem}

In~\cite{sinclair14:_lee_yang_theor_compl_comput_averag}, the theorem
was proved using a sequence of Asano-type
contractions~\cite{asano_lee-yang_1970}, a technique which originated
in Asano's proof of the Lee-Yang theorem~\cite{leeyan52b}.  The proof
we present here completely eschews the Asano contraction in favor of a
simpler analytic argument.  In our proof we need the following
version of the Lee-Yang theorem:
\begin{theorem}[\cite{leeyan52b,asano_lee-yang_1970}]\label{lyth}
  Let $G = (V, E)$ be a connected undirected graph on $n$ vertices,
  and suppose $0 < \beta < 1$.  Then $Z(G, \beta,
  \vec{z}) \neq 0$ if
  $\abs{z_v} \geq 1$ for all $v \in V$ and in addition 
  $\abs{z}_u > 1$ for some $u \in V$.  By symmetry, the conclusion also holds when
  $\abs{z_v} \leq 1$ for all $v \in V$ and in addition $\abs{z}_u < 1$ for some $u\in V$.
  %Let $\vec{z} = \nvec{z}$ be a collection of vertex activities such that
\end{theorem}

Observe that given any vertex $u \in V$, we can decompose the
partition function as
\begin{eqnarray}
  \label{eq:1}
Z(G, \beta, \vec{z}) & = & A \; z_{u} + B    \\\label{eq:11}
A = A(\vec{z}) & = & \beta^{\deg u} Z(G - \inb{u}, \beta, \vec{z}')\;\;\;\;\;\;\;\;\;        z'_{w} = \left\{\begin{array}{ll}
z_{w} & \mbox{when $w \not\sim u$ in $G$} \\
z_w/\beta & \mbox{when $w \sim u$ in $G$} 
\end{array}\right.          \\\nonumber
B = B(\vec{z}) & = & Z(G - \inb{u}, \beta, \vec{z}'')\;\;\;\;\;\;\;\;\;\;\;\;\;\;\;\;\;\;\;           z''_{w} = \left\{\begin{array}{ll}
z_{w} & \mbox{when $w \not\sim u$ in $G$} \\
z_w \beta & \mbox{when $w \sim u$ in $G$} 
\end{array}\right. 
\end{eqnarray}
Neither $\vec{z}'$ nor $\vec{z}''$ contains $z_{u}$ and $G - \inb{u}$ denotes the graph that we obtain from $G$ by leaving out node $u$.
The Lee-Yang theorem has the following simple consequence, which was also
used in~\cite{sinclair14:_lee_yang_theor_compl_comput_averag}.
\begin{lemma} \label{neqzero}
If $G$ is connected, $0 < \beta < 1$, and all vertex activities have absolute value 1,
then $A$ of eq. (\ref{eq:11}) is not zero. \end{lemma}
\begin{proof}
Since $\beta\neq 0$, it is sufficient to prove that $Z(G - \inb{u}, \beta, \vec{z}')\neq 0$.
We observe that the latter is a product of the partition functions of the connected
 components of $G-\inb{u}$, and furthermore, any neighbor $w$ of $u$ in
 $G$ in each such component has a vertex activity $z_{w}' =
  z_w/\beta$ with $|z_w/\beta | = |z_w| / \beta = 1/\beta >1$.  
  Due to $G$ being connected, we find such a neighbor $w$ of $u$ in all components of $G-\inb{u}$. We 
  apply Theorem \ref{lyth} to each 
  connected component of $G-\inb{u}$ separately to show that none of the factors
  is zero.
\end{proof}

\begin{proof}[Proof of Theorem~\ref{thm:1}]

  Let $G$ and $\beta$ be as in the hypotheses of the theorem. Suppose
  now that there exists a point $\vec{z^{0}}$ such that $\abs{z^{0}_v} = 1$
  for all $v$, and $\D{G}Z(G, \beta, \vec{z^{0}}) = 0$.  We will show
  that this leads to a contradiction. For our subsequent argument it will be helpful to 
  define the univariate polynomial
  \begin{displaymath}
    f(t) \defeq Z_G(G, \beta, t\vec{z^{0}})\;\;\;\;\;\;\mbox{where}\;\;\;\;\;\;\; t\vec{z^{0}} = \nvec{t z^{0}}
  \end{displaymath}
  
  \begin{lemma}
  $Z(G, \beta, \vec{z^{0}}) = 0$
  \end{lemma}
  \begin{proof} 
  A comparison of the individual terms 
  gives that $f'(1) = \D{G}Z(G, \beta, \vec{z^{0}})$, which is zero by our assumption.
  From the Lee-Yang theorem we obtain that $f(t) \neq 0$ when $\abs{t} \neq
  1$, so all zeros of $f$ must lie on the unit circle.  This together with the
  Gauss-Lucas lemma implies that the derivative of $f$ cannot disappear 
  on a point of the 
  unit circle 
  unless $f$ disappears at the same point.
  Thus, since  $f'(1) = 0$, we get that $Z(G, \beta, \vec{z^{0}}) =  f(1)= 0$.
 \end{proof}

  We have that 
  $
  f(1-\epsilon) = f(1) - \epsilon f'(1) \pm O(\epsilon^2) = \pm O(\epsilon^2), 
  $
  since 
  the first two terms are zero. Let $\vec{e_{u}}$ be the vertex activity (fugacity) vector
  with all zero vertex activities except at vertex $u$ that has activity $1$. The key to the proof is to 
  consider the linear perturbation
  \begin{equation}\label{pert}
  Z(G, \beta, (1-\epsilon) \vec{z^{0}} + \tau \vec{e_{u}}) 
  \end{equation}
  We show that (\ref{pert}) disappears for some $\tau\in \mathbb{C}$, $|\tau| < \epsilon $, in contradiction with the
  Lee-Yang theorem, since under this assumption all components of $(1-\epsilon)\vec{z^{0}} + \tau \vec{e_{u}}$
  have absolute value less than one. By (\ref{eq:1}):
  %\begin{equation}\label{xxxx}
  \[
  Z(G, \beta, (1-\epsilon) \vec{z^{0}} + \tau \vec{e_{u}}) = Z(G, \beta, (1-\epsilon) \vec{z^{0}}) + 
  A((1-\epsilon) \vec{z^{0}}) \tau =
 \mu_{\epsilon}  + A(\vec{z^{0}})\tau  +  \nu_{\epsilon}\tau 
 \]
 % \pm O(\epsilon^2) + \left( A(\vec{z^{0}}) +\pm O(\epsilon)\right) \tau =
 % \end{equation}
 Here $\mu_{\epsilon} = f(1-\epsilon) = \pm O(\epsilon^{2})$, and  $\nu_{\epsilon} =
 A((1-\epsilon) \vec{z^{0}}) - A(\vec{z^{0}}) =
 \pm O(\epsilon)$ by the analyticity of the function $A$.
  Recall that $A(\vec{z^{0}})\neq 0$ by Lemma \ref{neqzero}.
  Then expression (\ref{pert}) disappears at $\tau = -{\mu_{\epsilon}\over A(\vec{z^{0}})+\nu_{\epsilon}}$,
  and also $|\tau| < \epsilon$ if $\epsilon$ is sufficiently small.
  \end{proof}
  
  \section{Acknowledgements}
\label{sec:thanks}

The authors thank the Simons Institute (the Quantum Hamiltonian Complexity program).
The first author was supported by the Berkeley Fellow- ship for Graduate Study and NSF grant CCF-1016896
and the second by NSF Grant No. CCF-0832787, ``Understanding, Coping with, and Benefiting from, Intractability''
and by CISE/MPS  1246641.

\bibliographystyle{alpha}
\bibliography{proof}

\begin{thebibliography}{Asa70}

\bibitem[Asa70]{asano_lee-yang_1970}
Taro Asano.
\newblock {Lee-Yang} theorem and the {Griffiths} inequality for the anisotropic
  {Heisenberg} ferromagnet.
\newblock {\em Physical Review Letters}, 24(25):1409--1411, June 1970.

\bibitem[LY52]{leeyan52b}
T.~D. Lee and C.~N. Yang.
\newblock Statistical theory of equations of state and phase transitions. {II}.
  {Lattice} gas and {Ising} model.
\newblock {\em Physical Review}, 87(3):410--419, August 1952.

\bibitem[SS14]{sinclair14:_lee_yang_theor_compl_comput_averag}
Alistair Sinclair and Piyush Srivastava.
\newblock {Lee--Yang} theorems and the complexity of computing averages.
\newblock {\em Comm. Math. Phys.}, 329(3):827--858, 2014.

\end{thebibliography}

\end{document}